\documentclass[lettersize,journal]{IEEEtran}

\usepackage{amsmath,amssymb,amsfonts}
\usepackage{algorithmic}
\usepackage{algorithm}
\usepackage{graphicx}
\usepackage{textcomp}
\usepackage{xcolor}
\usepackage{cite}
\usepackage{bm}

\newtheorem{theorem}{Theorem}
\newtheorem{remark}{Remark}

\usepackage{balance}

\begin{document}

%\title{Pinching-Antenna Systems With Measured Radiation Patterns and Coupling Efficiency Optimization}
\title{Measured-Pattern-Aware Pinching-Antenna Systems With Coupling-Efficiency Optimization}

\author{Hao Feng, Hui Yang, Ming Zeng, Yulei Wang, Ebrahim Bedeer and Nian Xia

    \thanks{H. Feng is with Hunan Institute of Engineering, Xiangtan, China, and also with Donghua University, Shanghai, China (e-mail: 1219001@mail.dhu.edu.cn).}
    %\thanks{H. Feng is with School of Textile and Clothing, Hunan Institute of Engineering, Xiangtan, China, and also with the College of Textiles, Donghua University, Shanghai, China (e-mail: 1219001@mail.dhu.edu.cn).}

    \thanks{H. Yang is with Hunan Institute of Science and Technology, Yueyang, China (email: achelyal@163.com).}
    
    \thanks{M. Zeng is with Laval University, Quebec City, Canada (email: ming.zeng@gel.ulaval.ca).}

    \thanks{Y. Wang is with South-Central Minzu University, Wuhan, China (email: ylwang@mail.scuec.edu.cn).}

    \thanks{E. Bedeer is with University of Saskatchewan, Saskatoon, SK, Canada (email: e.bedeer@usask.ca).}

    \thanks{N. Xia is with Nanjing Normal University, Nanjing, China (e-mail: nian.xia@nnu.edu.cn).} 
    }

\maketitle

%\begin{abstract}
%Pinching-antenna (PA) systems are commonly analyzed with isotropic radiation models and fixed coupling efficiencies. However, practical PAs can exhibit strongly directional radiation patterns, and the coupling efficiency determines how guided power is extracted and distributed along the waveguide. This letter studies measurement-driven PA systems with explicit coupling-efficiency optimization. We develop a channel model that incorporates an externally obtained radiation pattern, waveguide attenuation, and sequential power extraction. For the single-PA case, we show that the optimal PA position maximizes the joint effect of directional gain, waveguide attenuation, and free-space path loss, and derive the minimum coupling efficiency required to outperform a fixed isotropic antenna. For the finite multi-PA case, we consider main-beam-aware and phase-matched PA placement. Under uniform coupling, we derive a one-dimensional optimality condition and show that, for similar effective channel amplitudes, the optimal coupling efficiency decreases with the number of PAs. For independently controllable coupling efficiencies, we obtain a closed-form solution in which the radiated power fraction is proportional to the directional channel power gain. Numerical results using a representative measured PA pattern validate the benefits of measurement-aware placement and coupling-efficiency optimization.
%\end{abstract}

\begin{abstract}
Pinching-antenna (PA) systems have been widely investigated as a flexible architecture for waveguide-enabled wireless transmission. Existing analytical models, however, often rely on isotropic radiation assumptions and simplified coupling-efficiency settings, which may overlook two practical design factors: the geometry-dependent radiation pattern of each PA and the sequential extraction of guided power along the waveguide. In this paper, we propose a measured-radiation-pattern-aware PA framework that incorporates an externally obtained radiation pattern, waveguide attenuation, and coupling-dependent power extraction. For a single PA, the resulting placement rule balances directional gain, waveguide loss, and free-space path loss, leading to a coupling-efficiency threshold for outperforming a fixed isotropic antenna. For multiple PAs, we study phase-matched placement and coupling-efficiency design under both uniform and independently controllable coupling. The uniform-coupling case yields a one-dimensional optimality condition and reveals that the preferred coupling efficiency decreases as more phase-matched PAs participate in coherent combining. The independently controllable case admits a closed-form power-allocation structure, where stronger effective directional channels receive larger radiated power fractions. Numerical results based on a representative measured PA radiation pattern demonstrate the importance of jointly accounting for measured-radiation-pattern-aware placement and coupling-efficiency optimization.
\end{abstract}


\section{Introduction}
Pinching-antenna (PA) systems have recently emerged as a flexible antenna architecture for future wireless networks \cite{Atsushi_22}. By placing small dielectric perturbations near a dielectric waveguide, part of the guided electromagnetic energy can be coupled out and radiated into free space. Since PAs can be activated at different locations along a pre-deployed waveguide, they provide new spatial degrees of freedom for improving line-of-sight connectivity and reducing access distance \cite{Yang_WCM25,liu2025pinching, zeng2025_WCM}.

Most existing analytical studies of PA systems rely on simplified radiation and coupling models. In particular, activated PAs are often modeled as isotropic radiators \cite{Tegos_2025, Zeng_COMML25,  fu2025}. Under this assumption, the PA-user channel is mainly determined by free-space path loss, waveguide attenuation, and phase coherence, and PA placement is commonly interpreted as a distance-dominated problem. Meanwhile, the coupling efficiency between the waveguide and the PA is often treated in a simplified manner. For example, it is commonly set to one in the single-PA case \cite{tyrovolas2025, zeng2025EETVT, hu2025}, assumed to produce equal power sharing in multi-PA systems \cite{ding2024, wang2024, Zhao_TCOM25}, or modeled through a fixed proportional extraction rule \cite{Wang_TCOM25}. However, the optimization of coupling efficiency for a given finite number of activated PAs has not been sufficiently investigated.

These two simplifications are closely related in practical PA systems. On the one hand, PA radiation can be highly directional and geometry-dependent \cite{Li_COMML26}. As a result, the closest PA to the user is not necessarily the best choice, since the user may lie outside its strong radiation direction. A farther PA may instead provide a stronger effective channel if the user is located near its main beam. On the other hand, once multiple directional PAs are activated along the same waveguide, their useful contributions depend not only on their directional channel gains but also on how the guided power is extracted and distributed among them. A large coupling efficiency benefits earlier PAs but leaves less residual power for downstream PAs, whereas a smaller coupling efficiency can support coherent combining across more PAs. Therefore, measured radiation patterns and coupling-efficiency design should be considered jointly rather than as separate issues.

Motivated by these observations, this paper studies directional PA systems with externally obtained radiation patterns and explicit coupling-efficiency optimization. The main contributions are summarized as follows:
\begin{itemize}
\item We develop a measured-radiation-pattern-aware PA channel model that incorporates directional radiation gain, waveguide attenuation, and sequential coupling loss. 
%The proposed model treats the radiation pattern as an external input and is therefore applicable to measured or simulated PA radiation patterns.
\item For the single-PA case, we derive the optimal location criterion and the minimum coupling efficiency required for the optimized directional PA to outperform a fixed isotropic antenna.
\item For the finite multi-PA case, we consider main-beam-aware and phase-matched placement, derive the optimality condition for uniform coupling, and obtain a closed-form solution for independently controllable coupling efficiencies.
\end{itemize}

\section{System Model and Single-PA Design}

\subsection{Measured-Radiation-Pattern-Aware PA Model}
%Consider a dielectric waveguide deployed along the $x$-axis, with the feed point at $x=0$ and signal propagation along the positive $x$ direction. A PA can be activated at any candidate location $x\in\mathcal{X}$, where $\mathcal{X}$ denotes the feasible activation region. The single-antenna user is located at ${\bf u}=(x_u,y_u)$, and the PA-user distance is
%\begin{equation}
%d(x)=\sqrt{(x-x_u)^2+y_u^2}.
%\end{equation}
%Let $\phi(x)$ be the departure direction from the PA to the user in the local PA coordinate system. The directional radiation gain toward direction $\phi$ is denoted by $G(\phi)$. In this work, $G(\phi)$ is treated as an external input obtained from simulation or measurement \cite{Li_COMML26}. Hence, the following model is not tied to a particular PA geometry and can be applied to any measured or simulated PA radiation pattern.
%Consider a dielectric waveguide deployed along the $x$-axis in a three-dimensional Cartesian coordinate system. The feed point is located at $(0,0,d_0)$, with $d_0$ denoting the waveguide height. 
Consider a dielectric waveguide deployed parallel to the $x$-axis in a three-dimensional Cartesian coordinate system. The feed point is located at $(0,0,d_0)$, where $d_0$ denotes the waveguide height, and the signal propagates along the positive $x$ direction.
A PA can be activated at any candidate waveguide coordinate $x\in\mathcal{X}$, where $\mathcal{X}$ denotes the feasible activation region, and its 3D position is denoted by ${\bf r}(x)=(x,0,d_0)$. 
%\begin{equation}
%{\bf r}(x)=(x,0,d_0).
%\label{eq:pa_position_3d}
%\end{equation}
The single-antenna user is located at ${\bf u}=(x_u,y_u,z_u)$. Hence, the PA-user distance is
\begin{equation}
d(x)=\left\|{\bf u}-{\bf r}(x)\right\|
=\sqrt{(x-x_u)^2+y_u^2+(z_u-d_0)^2}.
\label{eq:distance_3d}
\end{equation}

Denote the departure direction from the PA to the user by 
%Let ${\bf s}(x)=({\bf u}-{\bf r}(x))/d(x)$ denote the unit departure direction from the PA to the user. In the local PA coordinate system, this direction is represented by 
the azimuth-elevation pair $\boldsymbol{\Omega}(x)=(\varphi(x),\vartheta(x))$, given by
%where, for the coordinate convention in \eqref{eq:pa_position_3d},
\begin{subequations} \label{eq:direction_angles_3d}
\begin{align}
\varphi(x)&=\operatorname{atan2}\left(y_u,x_u-x\right),\\
\vartheta(x)&=\operatorname{atan2}\left(z_u-d_0,\sqrt{(x-x_u)^2+y_u^2}\right),
\end{align}
\end{subequations}
%where $\operatorname{atan2}(y,x)$ denotes the two-argument inverse tangent function that returns the angle of the vector $(x,y)$ with the correct quadrant.
where $\operatorname{atan2}$ denotes the two-argument inverse tangent function. 

The directional radiation gain toward $\boldsymbol{\Omega}(x)$ is denoted by $G(\boldsymbol{\Omega}(x))$. In this work, $G(\boldsymbol{\Omega}(x))$ is treated as an external input obtained from simulation or measurement \cite{Li_COMML26}. Hence, the following model is not tied to a particular PA geometry and can be applied to any measured or simulated 3D PA radiation pattern. If only a measured angular cut is available, $G(\boldsymbol{\Omega}(x))$ can be replaced by the corresponding gain value on that cut.

For a single PA with coupling efficiency $\rho\in(0,1]$, the channel coefficient is modeled as \cite{ding2024}
\begin{equation}
h(x;\rho)=
\sqrt{\rho G(\boldsymbol{\Omega}(x))}
\frac{\sqrt{\beta_0}e^{-\alpha_{\rm w}x}}{d(x)}
e^{-j\left(\frac{2\pi}{\lambda}d(x)+\frac{2\pi}{\lambda_g}x\right)},
\label{eq:single_channel_loss}
\end{equation}
where $\beta_0=\lambda^2/(16\pi^2)$, $\lambda$ and $\lambda_g$ are the free-space and guided wavelengths, respectively, and $\alpha_{\rm w}$ is the waveguide amplitude attenuation coefficient. The phase term includes both the free-space propagation phase and the guided-wave phase accumulated from the feed point to the PA. For transmit power $P$ and noise power $\sigma^2$, the received SNR is
\begin{equation}
\gamma(x;\rho)=
\frac{P\beta_0\rho}{\sigma^2}
\frac{G(\boldsymbol{\Omega}(x))e^{-2\alpha_{\rm w}x}}{d^2(x)} .
\label{eq:single_snr_loss}
\end{equation}

\subsection{Single-PA Placement and Coupling Threshold}
From \eqref{eq:single_snr_loss}, the optimal radiation pattern-aware PA location is
\begin{equation}
x^\star=\arg\max_{x\in\mathcal{X}}
\frac{G(\boldsymbol{\Omega}(x))e^{-2\alpha_{\rm w}x}}{d^2(x)},
\label{eq:directional_location_loss}
\end{equation}
which can be found through a one-dimensional search. This criterion differs from conventional distance-based placement. A farther PA may outperform the closest one if it provides a sufficiently stronger directional gain toward the user while keeping moderate waveguide attenuation.

We next compare the optimized directional PA with a conventional fixed isotropic antenna located at ${\bf b}$. Let $d_{\rm F}=\|{\bf u}-{\bf b}\|$. The fixed-antenna SNR is $\gamma_{\rm F}=P\beta_0/(\sigma^2d_{\rm F}^2)$. Therefore, the optimized PA outperforms the fixed isotropic antenna when
\begin{equation}
\rho>\rho_{\min}
=
\frac{1}{
d_{\rm F}^2
\displaystyle
\max_{x\in\mathcal{X}}
\frac{G(\boldsymbol{\Omega}(x))e^{-2\alpha_{\rm w}x}}{d^2(x)}
}.
\label{eq:rho_min_loss}
\end{equation}
If $\rho_{\min}>1$, even perfect coupling cannot make the directional PA outperform the conventional fixed isotropic antenna under the considered geometry, waveguide attenuation, and radiation pattern.

\section{Finite Multi-PA System}
Consider $N$ activated PAs at $x_1<\cdots<x_N$. For the $n$-th PA, denote its distance and departure direction to the user by $d_n=d(x_n)$ and $\boldsymbol{\Omega}_n=\boldsymbol{\Omega}(x_n)$, respectively.  The accumulated propagation phase of the $n$-th PA-user link is
\begin{equation}
\psi_n=\frac{2\pi}{\lambda}d_n+\frac{2\pi}{\lambda_g}x_n.
\label{eq:multi_def}
\end{equation}

Let $\rho_n$ be the local coupling efficiency. Due to sequential power extraction along the waveguide, the fraction of input guided power radiated by the $n$-th PA is
\begin{equation}
p_n=\rho_n\prod_{i=1}^{n-1}(1-\rho_i).
\label{eq:pn_individual}
\end{equation}
The corresponding channel contribution is
\begin{equation}
h_n(p_n)=
\sqrt{p_nG(\boldsymbol{\Omega}_n)}
\frac{\sqrt{\beta_0}e^{-\alpha_{\rm w}x_n}}{d_n}e^{-j\psi_n}.
\label{eq:hn_uniform_loss}
\end{equation}
Define the effective nonnegative channel amplitude as
\begin{equation}
a_n=\frac{\sqrt{G(\boldsymbol{\Omega}_n)}e^{-\alpha_{\rm w}x_n}}{d_n}.
\label{eq:an_loss}
\end{equation}
Then the received SNR can be written compactly as
\begin{equation}
\gamma_N=
\frac{P\beta_0}{\sigma^2}
\left|
\sum_{n=1}^{N}\sqrt{p_n}a_ne^{-j\psi_n}
\right|^2 .
\label{eq:snr_general}
\end{equation}

To obtain coherent combining without active per-PA phase control, the PAs are selected near the favorable region determined by \eqref{eq:directional_location_loss} while satisfying the phase-matching condition \cite{ding2024}
\begin{equation}
\psi_n-\psi_1=2\pi k_n,\quad k_n\in\mathbb{Z},\quad n=2,\ldots,N,
\label{eq:phase_matching_condition}
\end{equation}
or equivalently,
\begin{equation}
\frac{d_n-d_1}{\lambda}+\frac{x_n-x_1}{\lambda_g}=k_n .
\label{eq:phase_matching_condition_normalized}
\end{equation}
In this case, \eqref{eq:snr_general} reduces to the coherent-combining form
\begin{equation}
\gamma_N=
\frac{P\beta_0}{\sigma^2}
\left(
\sum_{n=1}^{N}\sqrt{p_n}a_n
\right)^2 .
\label{eq:snr_phase_matched}
\end{equation}
This expression shows that multi-PA performance depends on both the effective amplitudes $\{a_n\}$, which include the measured radiation pattern $G(\boldsymbol{\Omega}_n)$, and the radiated power fractions $\{p_n\}$, which are controlled by the coupling efficiencies $\rho_n$.

\section{Coupling-Efficiency Optimization}

\subsection{Uniform Coupling}
We first consider the uniform-coupling case $\rho_n=\rho$, $\forall n$, which is relevant when identical PA elements sharing the same coupling design are deployed. In this case,
\begin{equation}
p_n =\rho(1-\rho)^{n-1}.
\end{equation}
Under phase matching, the optimal uniform coupling efficiency is
\begin{equation}
\rho^\star=\arg\max_{0<\rho\leq1}
\left(
\sum_{n=1}^{N}\sqrt{\rho(1-\rho)^{n-1}}a_n
\right)^2 .
\label{eq:rho_star_uniform}
\end{equation}
Let $t=\sqrt{1-\rho}$ and $B(t)=\sum_{n=1}^{N}a_nt^{n-1}$. Then \eqref{eq:rho_star_uniform} is equivalent to maximizing $(1-t^2)B^2(t)$ over $0\leq t<1$. The first-order optimality condition is
\begin{equation}
(1-t^2)B'(t)=tB(t),
\label{eq:opt_condition_compact}
\end{equation}
or, equivalently,
\begin{equation}
(1-t^2)
\sum_{n=2}^{N}(n-1)a_nt^{n-2}
=
t\sum_{n=1}^{N}a_nt^{n-1}.
\label{eq:opt_condition_expand}
\end{equation}
The optimal coupling efficiency is then $\rho^\star=1-(t^\star)^2$, where $t^\star$ is the solution of \eqref{eq:opt_condition_expand} that maximizes $(1-t^2)B^2(t)$. Since this is a one-dimensional problem, it can be solved efficiently by grid search.

\subsection{Similar Effective Channel Amplitudes}
To obtain further insight, consider the case where the phase-matched PAs are located in a small favorable region. Their directional gains, distances, and waveguide losses are then comparable, so that $a_1\approx\cdots\approx a_N=a$. In this case,
\begin{equation}
B(t)=a\frac{1-t^N}{1-t},
\end{equation}
and \eqref{eq:opt_condition_compact} reduces to
\begin{equation}
1-t^N=Nt^{N-1}(1-t^2).
\label{eq:special_condition}
\end{equation}
The optimal coupling efficiency is
\begin{equation}
\rho_N^\star=1-(t_N^\star)^2,
\label{eq:special_rho}
\end{equation}
where $t_N^\star$ is the root of \eqref{eq:special_condition}.

\begin{theorem}
For $N\geq2$, $t_N^\star$ is unique in $(0,1)$ and monotonically increases with $N$. Hence, $\rho_N^\star$ decreases with $N$.
\end{theorem}

\begin{IEEEproof}
Rewrite \eqref{eq:special_condition} as
\begin{equation}
\sum_{k=0}^{N-1}t^k=Nt^{N-1}(1+t).
\end{equation}
Define
\begin{equation}
q_N(t)
=
Nt^{N-1}(1+t)-\sum_{k=0}^{N-1}t^k=Nt^N+(N-1)t^{N-1}-\sum_{k=0}^{N-2}t^k .
\end{equation}
Then $t_N^\star$ is the unique root of $q_N(t)=0$ in $(0,1)$ for $N\geq2$. This uniqueness follows from the Descartes’ rule of signs \cite{anderson1998descartes} because $q_N(t)$ has exactly one sign change in its polynomial coefficients, while $q_N(0)<0$ and $q_N(1)>0$.

Next, evaluating $q_{N+1}(t)$ at $t=t_N^\star$ gives
\begin{subequations}
\begin{align}
    q_{N+1}(t_N^\star)
    &=
    (N+1)(t_N^\star)^N(1+t_N^\star)
    -\sum_{k=0}^{N}(t_N^\star)^k \\ 
    &=
    (t_N^\star)^{N-1}
    \left((N+1)(t_N^\star)^2-N\right).
\end{align} \label{eq:qNplus_eval}
\end{subequations}
    
We now show that $q_{N+1}(t_N^\star)<0$. From the optimality condition \eqref{eq:opt_condition_compact}, we have
\begin{equation}
\frac{t_N^\star B_N'(t_N^\star)}
{B_N(t_N^\star)}
=
\frac{(t_N^\star)^2}{1-(t_N^\star)^2}.
\end{equation}
Since
\begin{equation}
\frac{t B_N'(t)}{B_N(t)}
=
\frac{\sum_{k=0}^{N-1}k t^k}
{\sum_{k=0}^{N-1}t^k} < N-1,
\end{equation}
we have
\begin{equation}
\frac{(t_N^\star)^2}{1-(t_N^\star)^2}<N-1,
\end{equation}
which implies
\begin{equation}
(t_N^\star)^2<\frac{N-1}{N}<\frac{N}{N+1}.
\end{equation}
Substituting this inequality into \eqref{eq:qNplus_eval} yields
\begin{equation}
q_{N+1}(t_N^\star)<0.
\end{equation}
Since $q_{N+1}(t)$ has a unique root in $(0,1)$ and $q_{N+1}(1)>0$, its root must satisfy
\begin{equation}
t_{N+1}^\star>t_N^\star.
\end{equation}
Therefore, $t_N^\star$ is monotonically increasing with $N$, and thus, $\rho_N^\star$ decreases with $N$ based on \eqref{eq:special_rho}.
\end{IEEEproof}

\begin{remark}
The unique root $t_N^\star$ can be obtained by bisection with complexity $\mathcal{O}(\log(1/\epsilon))$, where $\epsilon$ is the prescribed numerical tolerance. Since $t_N^\star$ increases with $N$, a smaller uniform coupling efficiency is preferred when more phase-matched PAs participate in coherent combining.
\end{remark}

\subsection{Independently Controllable Coupling}
We now allow each PA to have an independently controlled coupling efficiency. Since residual power after the last PA does not contribute to the received signal, the optimal radiated fractions satisfy
\begin{equation}
    \sum_{n=1}^{N}p_n=1. 
    \label{eq:sum_power_2}
\end{equation}

Under phase matching, maximizing the SNR is equivalent to
\begin{equation}
\max_{\{p_n\}}
\sum_{n=1}^{N}a_n\sqrt{p_n},
\quad
{\rm s.t.}\quad
\sum_{n=1}^{N}p_n=1,\quad p_n\geq0.
\label{eq:power_fraction_opt}
\end{equation}
By the Cauchy-Schwarz inequality,
\begin{equation}
    \left(
    \sum_{n=1}^{N}a_n\sqrt{p_n}
    \right)^2
    \leq
    \left(
    \sum_{n=1}^{N}a_n^2
    \right)
    \left(
    \sum_{n=1}^{N}p_n
    \right).
\end{equation}
Using \eqref{eq:sum_power_2}, the upper bound is
\begin{equation}
    \left(
    \sum_{n=1}^{N}a_n\sqrt{p_n}
    \right)^2
    \leq
    \sum_{n=1}^{N}a_n^2.
\end{equation}
The equality holds when
%\begin{equation}
%    \sqrt{p_n^\star}
%    =
%    \frac{a_n}{\sqrt{\sum_{i=1}^{N}a_i^2}},
%\end{equation}
%or equivalently,
\begin{equation}
    p_n^\star
    =
    \frac{a_n^2}{\sum_{i=1}^{N}a_i^2}.
    \label{eq:pn_star}
\end{equation}

Thus, PAs with stronger directional gain, smaller waveguide attenuation, and shorter PA-user distance radiate larger fractions of the input power. From \eqref{eq:pn_individual}, the corresponding local coupling efficiency is
\begin{equation}
\rho_n^\star
=
\frac{p_n^\star}{1-\sum_{i=1}^{n-1}p_i^\star}
=
\frac{a_n^2}{\sum_{i=n}^{N}a_i^2},
\label{eq:rho_n_star}
\end{equation}
and $\rho_N^\star=1$, i.e., the last PA radiates all remaining guided power. The resulting maximum SNR is
\begin{equation}
\gamma_N^\star=
\frac{P\beta_0}{\sigma^2}
\sum_{n=1}^{N}a_n^2 .
\label{eq:max_snr_individual}
\end{equation}

\begin{figure}[t]
\centering
{\includegraphics[width=0.5\textwidth]{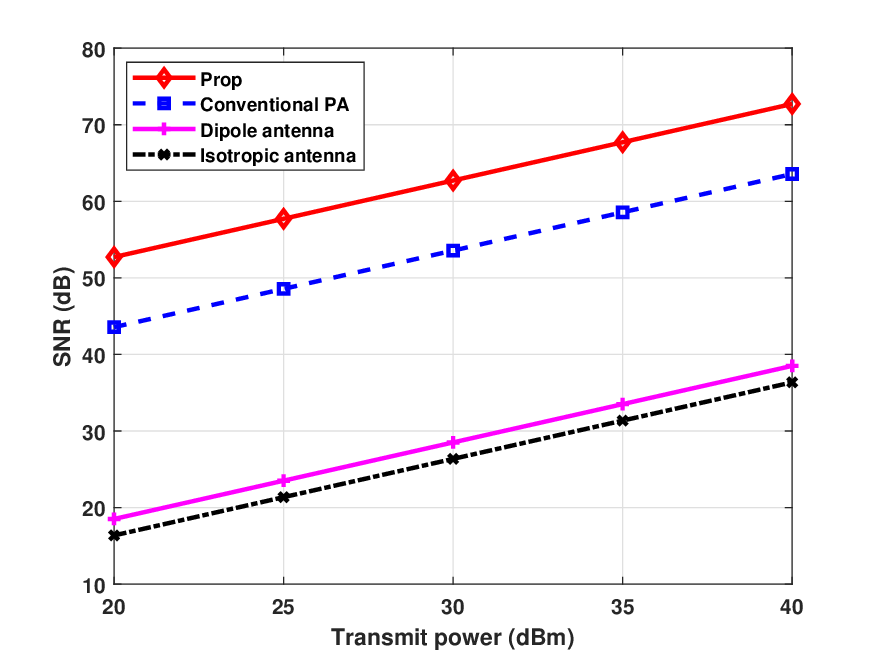}}
\caption{SNR (dB) versus transmit power.}
\label{fig:power}
\end{figure}

\begin{figure}[t]
\centering
{\includegraphics[width=0.5\textwidth]{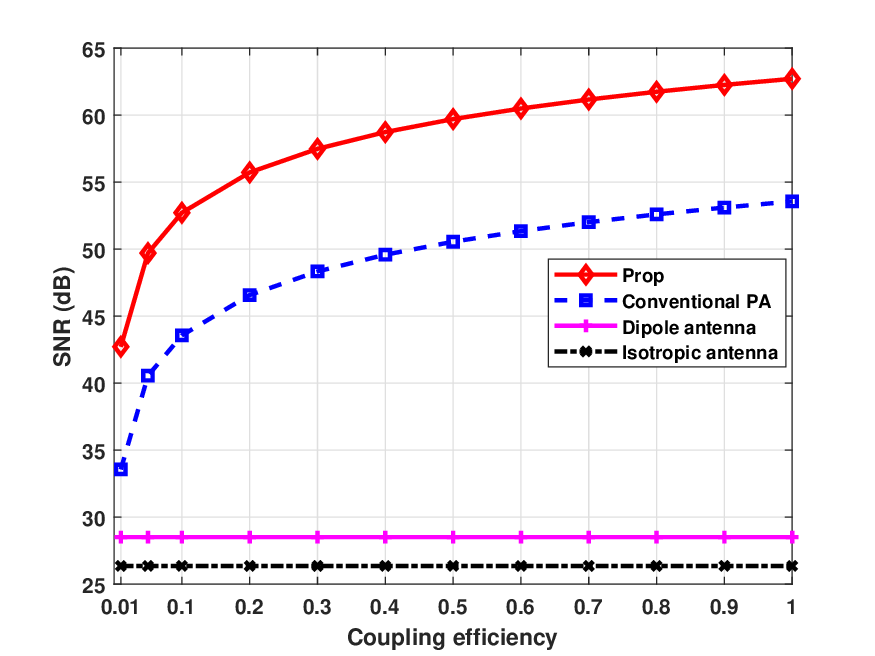}}
\caption{SNR (dB) versus coupling efficiency $\rho$.}
\label{rho_single}
\end{figure}

\begin{figure}[t]
\centering
{\includegraphics[width=0.5\textwidth]{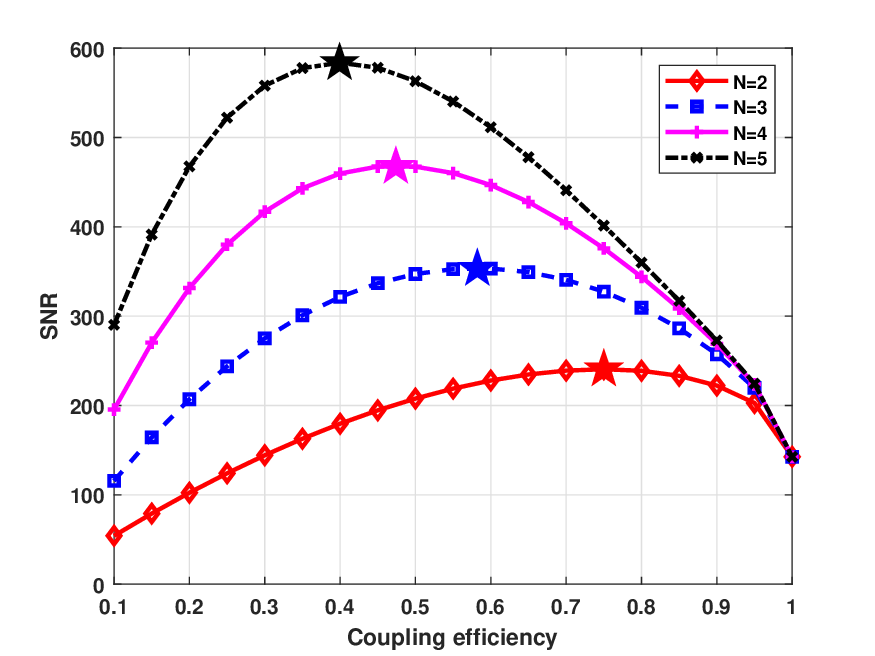}}
\caption{SNR versus coupling efficiency $\rho$.}
\label{rho_multi}
\end{figure}

\section{Numerical Results}
This section validates the analytical results and highlights the impact of measured radiation patterns and coupling-efficiency design. The radiation pattern is instantiated using the measured square-PA radiation-pattern data reported in \cite{Li_COMML26}, where the PA and user are assumed to have the same height. Unless otherwise stated, the carrier frequency is $60$ GHz, the waveguide length is $100$ m, and the user is randomly located in a $100\times 40$ m$^2$ rectangular area. The waveguide attenuation coefficient is $\alpha_{\rm w}=0.005$ \cite{tyrovolas2025}, the default transmit power is $30$ dBm, and the noise power is $-94$ dBm. All results are averaged over $10^3$ random realizations. 

Figs.~\ref{fig:power} and \ref{rho_single} compare the proposed measured-radiation-pattern-aware PA scheme with a conventional PA placement benchmark, where the PA position is optimized only according to waveguide attenuation and free-space path loss.
For fairness, after the benchmark location is selected, its SNR is evaluated using the same externally obtained measured square-PA radiation pattern as the proposed scheme. Hence, the performance gap reflects the benefit of using the measured radiation pattern in the placement decision. Two fixed-antenna benchmarks, namely isotropic and dipole antennas, are also included and assumed to be located at the feeding point. 
%Figs.~\ref{fig:power} and \ref{rho_single} compare the proposed measurement-aware PA scheme with a conventional PA benchmark, where the PA position is optimized only according to waveguide attenuation and free-space path loss. Two fixed-antenna benchmarks, namely isotropic and dipole antennas, are also included. 
As shown in Fig.~\ref{fig:power}, the SNR increases with the transmit power for all schemes. The proposed scheme achieves a clear gain over the conventional PA benchmark, which confirms that the measured radiation pattern should be included in PA placement. The PA-based schemes also outperform the fixed-antenna benchmarks because of their flexible placement and directional radiation gain.

Fig.~\ref{rho_single} shows the SNR versus the coupling efficiency for the single-PA case. The SNR of the PA schemes increases with $\rho$ because more guided power is radiated by the activated PA. Under the considered setup, the optimized directional PA dominates the fixed-antenna benchmarks, showing that coupling efficiency and radiation direction jointly affect the achievable SNR.

Fig.~\ref{rho_multi} shows the SNR versus the uniform coupling efficiency for different numbers of PAs. The stars mark the optimal coupling efficiencies predicted by \eqref{eq:special_rho}. For each $N$, the SNR first increases and then decreases with $\rho$, reflecting the tradeoff between extracting power at earlier PAs and preserving guided power for downstream coherent combining. Moreover, the optimal $\rho$ decreases as $N$ increases, which verifies Theorem~1.

\section{Conclusion}
This paper studied measured-radiation-pattern-aware PA systems with explicit coupling-efficiency optimization. By incorporating directional radiation gain, waveguide attenuation, and sequential power extraction, we showed that PA placement should not be determined by distance alone. For the single-PA case, we derived the optimal placement criterion and the coupling-efficiency threshold required to outperform a fixed isotropic antenna. For the finite multi-PA case, we considered main-beam-aware and phase-matched placement. Under uniform coupling, the optimal coupling efficiency was characterized by a one-dimensional condition, and for similar effective channel amplitudes, it was shown to decrease with the number of activated PAs. For independently controllable coupling efficiencies, a closed-form solution was obtained, where the optimal radiated power fraction is proportional to the directional channel power gain. These results demonstrate that measured radiation patterns and coupling-efficiency design are both important for practical PA system optimization.

\section{Acknowledgment}
We would like to express our gratitude to the authors of \cite{Li_COMML26} for sharing their raw measurement data with us. 

\bibliographystyle{IEEEtran}
\bibliography{biblio}

\balance

\end{document}